\documentclass[letterpaper, 10 pt, conference]{ieeeconf}  

\IEEEoverridecommandlockouts                              

\usepackage{graphicx,algorithm,algorithmic,amsthm,amsmath,amsfonts,verbatim,mathtools,enumerate,bm,hyperref,cite,wrapfig,subcaption}
\usepackage[framemethod=TikZ]{mdframed}
\usepackage{pgfplots}
\usepackage{tikz}
\usepgfplotslibrary{fillbetween}
\pgfplotsset{compat=1.18}
\usepackage{subcaption}

\allowdisplaybreaks

\newcommand{\argmin}{\mathop{\rm argmin}}
\newcommand{\argmax}{\mathop{\rm argmax}}

\newcommand{\mnorm}[1]{{\left\vert\kern-0.25ex\left\vert\kern-0.25ex\left\vert #1 
    \right\vert\kern-0.25ex\right\vert\kern-0.25ex\right\vert}}

\newtheorem{definition}{Definition} 
\newtheorem{theorem}{Theorem}

\newcommand{\eg}{{\it e.g.}}

\newtheorem*{probdef*}{Problem Definition}

\newtheorem*{aspdef*}{Data Assumption}

\title{\LARGE \bf Actively Learning to Coordinate in Convex Games via Approximate Correlated Equilibrium  
}

\author{Zhenlong~Fang, Aryan~Deshwal, and Yue~Yu
\thanks{Z. Fang and A. Deshwal are with the Department of Computer Science and Engineering, University of Minnesota Twin Cities, Minneapolis, MN 55455 (emails:  fang0540@umn.edu,\,adeshwal@umn.edu). Y. Yu is with the Department of Aerospace Engineering and Mechanics, University of Minnesota Twin Cities, Minneapolis, MN 55455 (email: yuey@umn.edu).
}
}

\begin{document}

\maketitle
\thispagestyle{empty}
\pagestyle{empty}

\begin{abstract}
Correlated equilibrium generalizes Nash equilibrium by allowing a central coordinator to guide players’ actions through shared recommendations, similar to how routing apps guide drivers. We investigate how a coordinator can learn a correlated equilibrium in convex games where each player minimizes a convex cost function that depends on other players’ actions, subject to convex constraints without knowledge of the players’ cost functions. We propose a learning framework that learns an approximate correlated equilibrium by actively querying players’ regrets, \emph{i.e.}, the cost saved by deviating from the coordinator’s recommendations. We first show that a correlated equilibrium in convex games corresponds to a joint action distribution over an infinite joint action space that minimizes all players’ regrets. To make the learning problem tractable, we introduce a heuristic that selects finitely many representative joint actions by maximizing their pairwise differences. We then apply Bayesian optimization to learn a probability distribution over the selected joint actions by querying all players' regrets. The learned distribution approximates a correlated equilibrium by minimizing players' regrets. We demonstrate the proposed approach via numerical experiments on multi-user traffic assignment games in a shared transportation network.
\end{abstract}

\section{Introduction}

Correlated equilibrium is game-theoretic solution concept that generalizes Nash equilibrium \cite{aumann1974subjectivity,aumann1987correlated,hart1989existence,cavaliere2001coordination,babichenko2014simple}. It applies to settings where players can receive recommendations from a central coordinator---such as routing apps---that help them coordinate their actions. A correlated equilibrium is a probability distribution over all players' joint actions that, when joint actions are sampled from this distribution and recommended to the players, no player has incentives to unilaterally deviate from the recommendation. By allowing players to correlate their actions via shared recommendations, correlated equilibrium can lead to outcomes that are more efficient and fair, especially when players compete for resources \cite{duffy2017coordination,wu2012cooperative}. 

Computing correlated equilibrium is challenging, as it requires full knowledge of all players' cost functions and consideration of all possible joint actions. For example, in the case of multiplayer matrix games where each player has a finite set of actions, computing a correlated equilibrium is equivalent to solve a linear program \cite{papadimitriou2008computing,wu2012cooperative}. In this linear program, the number of variables equals the total number of possible joint actions of all players, and the constraints depend on each player's cost functions. As the number of players and actions grows, solving this linear program becomes computationally intractable. Consequently, some results propose approximating the correlated equilibrium conditions in graphical games, leading to near-equilibrium solutions in randomly generated games \cite{kamisetty2011approximating}. Recent results further showed that one can compute special cases of correlated equilibria in multiplayer matrix games by computing multiple distinct Nash equilibria, which is more efficient than solving the linear program as the problem size grows \cite{im2024coordination}. 

Learning-based methods provide a decentralized approach to computing correlated equilibria \cite{hart2000simple,foster1997calibrated}. These methods design decision rules that guide players to update their actions in repeated games, using the observed costs of past actions. The goal is for the empirical distribution of the players' joint actions to converge to a correlated equilibrium. For example, the regret matching method adjusts each player's action based on the regret of not having chosen other actions in the past \cite{hart2000simple,wu2012cooperative}. While regret matching ensures convergence to a correlated equilibrium, it does not necessarily lead to an efficient correlated equilibrium that minimize players’ total costs. To address this limitation, later approaches proposed trial-and-error learning, where players experiment with new actions depending on whether they are content with baseline actions \cite{borowski2014learning,marden2017selecting}. Recent work explores similar ideas but allowing players to update their actions using only a limited memory of past interactions \cite{arifovic2019learning}. 

A common assumption in existing learning methods for correlated equilibrium is that each player independently adjusts their actions without external coordination \cite{foster1997calibrated,hart2000simple,borowski2014learning,marden2017selecting,arifovic2019learning}. While this feature supports decentralized learning, many real-world systems naturally include centralized coordination---such as routing apps---that guide players’ decisions. As a result, while existing methods allow players to independently adapt toward a correlated equilibrium, they offer limited insight into \emph{how a coordinator can efficiently learn such an equilibrium from players’ feedback}. Addressing this question is critical for designing centralized coordination mechanisms---such as routing apps that learn from user feedback---that relieve players of the burden of adjusting their actions through uncoordinated interactions and promote efficient joint decision-making..

Another limitation of the existing methods, for both computing \cite{papadimitriou2008computing,kamisetty2011approximating,wu2012cooperative,im2024coordination} and learning correlated equilibrium \cite{foster1997calibrated,hart2000simple,borowski2014learning,marden2017selecting,arifovic2019learning}, is that they focus on games with finite action sets. In practice, many common decision-making models---such as Markov decision processes, shortest path problems, and linear quadratic regulators---represent each player's decision-making as an optimization over a continuous action space. Consequently, learning a correlated equilibrium requires new frameworks to handle the challenges posed by infinitely many joint actions.

We propose a active learning framework in which a coordinator actively learns an approximate correlated equilibrium in convex games. In these games, each player minimizes a cost function that is convex in its own decision variables, subject to convex constraints. To make the learning over infinitely many joint actions tractable, we introduce the notion of approximate correlated equilibria---joint action distributions with finite support over a set of basis joint actions. The coordinator efficiently learns such an equilibrium by interacting with the players through a sequence of queries. In each query, it presents a joint action distribution and collects the players’ \emph{correlated regrets}, which measure the cost players can save by unilaterally deviating from recommended actions.

To this end, we develop algorithmic tools for selecting and updating joint action distributions that support active and query-efficient learning, and evaluate their effectiveness in coordinating traffic assignments. We first propose a Bayesian optimization method that adjusts the weights on the basis actions by querying players' regrets, balancing exploration and exploitation. To identify a representative set of basis actions, we propose a heuristic that selects joint actions by maximizing their pairwise differences. Finally, we demonstrate the proposed learning framework in multi-user traffic assignment games.

\section{Correlated Equilibrium in Convex Games}
We begin our discussions with convex games and correlated equilibrium. We further introduce some novel concepts such as correlated regret and theor variations. These concepts underpin our later discussions.
\subsection{Convex games}
We consider a game with \(m\in\mathbb{N}\) players. Each player's goal is to minimize a cost function, whose value depends on the other players' actions, subject to convex constraints. We let \(\mathbb{X}_i\subset\mathbb{R}^{n_i}\) denote the set of actions for player \(i\), and 
\begin{equation}
    \mathbb{X}\coloneqq \prod_{i=1}^m \mathbb{X}_i, \enskip \mathbb{X}_{-i}\coloneqq \prod_{j=1, j\neq i}^m \mathbb{X}_j
\end{equation}
denote the set of joint actions of all players and the set of actions of players other than player \(i\), respectively. We let 
\begin{equation}
    f_i:\mathbb{X}_i\times \mathbb{X}_{-i}\to\mathbb{R}
\end{equation}
denote the cost function for player \(i\). In particular, \(f_i(x_i, x_{-i})\) denote the cost of action \(x_i\in\mathbb{X}_i\) for player \(i\), assuming other players take action \(x_{-i}\in\mathbb{X}_{-i}\). With these notations, we define a convex game as follows.

\begin{definition}[Convex Games]
    A tuple \(\{\mathbb{X}_i,f_i\}_{i=1}^m\) defines a convex game with \(m\) players (\(m\in\mathbb{N}\)) if, for each player \(i\), \(\mathbb{X}_i\in\mathbb{R}^n\) is a closed convex set, and \(f_i(x_i, x_{-i})\) is a convex function of \(x_i\) for any \(x_{-i}\in\mathbb{X}_{-i}\).  
\end{definition}

Convex games model a wide range of noncooperative multiplayer decision-making problems. For example, when each player's decision constraints arise from structured problems such as shortest path problems, linear quadratic regulator problems, or Markov decision processes, their interactions can be modeled as a convex game, as long as each player's cost function is convex in its own decision.

\subsection{Correlated equilibrium}
Correlated equilibrium generalizes the Nash equilibrium via recommendations from a coordinator. A joint action distribution is a correlated equilibrium if no player can reduce its expected cost by unilaterally deviating from the recommended action, assuming others follow the recommendations. Below, we define the correlated equilibrium in a convex game.

\begin{definition}[Correlated Equilibrium]\label{def: CE}
A probability measure \(\mu:\mathbb{X}\to\mathbb{R}_{\geq 0}\) is a \emph{correlated equilibrium} of convex game \(\{\mathbb{X}_i,f_i\}_{i=1}^m\) if, for each player \(i\) and any measurable function \(\sigma_i \colon \mathbb{X}_i \to \mathbb{X}_i\), the following condition holds:
\begin{equation}\label{eqn: CE general}
    \int_{\mathbb{X}} \left( f_i(x_i,x_{-i})-f_i\bigl(\sigma_i(x_i),x_{-i}\bigr) \right)\,d\mu(x_i,x_{-i}) \leq 0.
\end{equation}
\end{definition}

Notice that directly verifying whether a probability measure \(\mu\) is a correlated equilibrium using Definition~\ref{def: CE} requires verifying infinitely many instances of inequality \eqref{eqn: CE general}, one for each measurable function per player. To develop a computationally tractable test for correlated equilibrium, we now introduce the notion of \emph{correlated regret}.   
\begin{definition}[Correlated Regret]
    Let \(\{\mathbb{X}_i,f_i\}_{i=1}^m\) be a convex game with \(m\in\mathbb{N}\) players. Let \(\mu:\mathbb{X}\to\mathbb{R}_{\geq 0}\) be a probability measure on \(\mathbb{X}\). The \emph{correlated regret} of player \(i\) associated with \(\mu\), denoted by \(r_i(\mu)\in\mathbb{R}\), is defined as follows:
    \begin{equation}\label{eqn: correlated regret}
    \begin{aligned}
        r_i(\mu) \coloneqq &\int_{\mathbb{X}} f_i(x_i,x_{-i}) \,d\mu(x_i,x_{-i})\\
        &-\underset{y_i\in\mathbb{X}_i}{\min}\int_{\mathbb{X}}f_i(y_i,x_{-i}) \,d\mu(x_i,x_{-i}). 
    \end{aligned}
    \end{equation}
\end{definition}

Intuitively, the correlated regret measures the maximum amount of cost that could be saved by not following the actions recommended according to probability measure \(\mu\). 

The following theorem shows that one can verify a joint action distribution \(\mu\) is a correlated equilibrium by evaluating each player's correlated regret. 
\begin{theorem}\label{thm: equivalence-cost}
Let \(\{\mathbb{X}_i,f_i\}_{i=1}^m\) be a convex game with \(m\in\mathbb{N}\) players. A probability measure \(\mu:\mathbb{X}\to\mathbb{R}_{\geq 0}\) is a correlated equilibrium of convex game \(\{\mathbb{X}_i,f_i\}_{i=1}^m\) if and only if \(r_i(\mu)\leq 0\) for all \(i=1, 2, \ldots, m\).
\end{theorem}

\begin{proof}
Suppose that \eqref{eqn: CE general} holds for any measurable function \(\sigma_i\colon \mathbb{X}_i \to \mathbb{X}_i\). Letting \(y_i\in \mathbb{X}_i\) and \(\sigma_i(x_i) \equiv y_i\) in \eqref{eqn: CE general} gives
\begin{equation}
     \int_{\mathbb{X}} (f_i(x_i,x_{-i})-f_i(y_i,x_{-i})) \,d\mu(x_i,x_{-i})\leq 0.
\end{equation}
Since the above condition holds for any \(y_i\in\mathbb{X}_i\), we conclude that \(r_i(\mu)\leq 0\).

On the other hand, suppose that \(r_i(\mu)\leq 0\). let \(\sigma_i\colon \mathbb{X}_i\to \mathbb{X}_i\) be an arbitrary measurable function. For each fixed \(x_{-i}\in \mathbb{X}_{-i}\), let \(\mu(\cdot\mid x_{-i})\) denote the conditional probability measure on \(\mathbb{X}_i\) induced by \(\mu\). Let \(\bar{x}_i(x_{-i})\) be the conditional expectation such that
\[
\bar{x}_i(x_{-i}) := \mathbb{E}_{\mu(\cdot\mid x_{-i})}\bigl[\sigma_i(x_i)\bigr] = \int_{\mathbb{X}_i} \sigma_i(x_i)\,d\mu(x_i\mid x_{-i}).
\]
Since \(\mathbb{X}_i\) is closed and convex, we have that \(\bar{x}_i(x_{-i}) \in \mathbb{X}_i\) for each \(x_{-i}\).
Furthermore, since \(f_i(\cdot,x_{-i})\) is convex, by using Jensen's inequality we can show, for each fixed \(x_{-i}\),
\begin{equation}
    \begin{aligned}
    &\int_{\mathbb{X}_i} f_i\Bigl(\bar{x}_i(x_{-i}),x_{-i}\Bigr) \,d\mu(x_i\mid x_{-i})\\
        &=f_i\Bigl(\bar{x}_i(x_{-i}),x_{-i}\Bigr) \leq \int_{\mathbb{X}_i} f_i\bigl(\sigma_i(x_i),x_{-i}\bigr)\,d\mu(x_i\mid x_{-i}).
    \end{aligned}
\end{equation}
Integrating both sides with respect to the marginal measure of \(\mu\) on \(\mathbb{X}_{-i}\), we obtain
\begin{equation}\label{eqn: pf eqn1}
\begin{aligned}
   & \int_{\mathbb{X}} f_i\Bigl(\bar{x}_i(x_{-i}),x_{-i}\Bigr)\,d\mu(x_i,x_{-i})
\\
&\le \int_{\mathbb{X}} f_i\bigl(\sigma_i(x_i),x_{-i}\bigr)\,d\mu(x_i,x_{-i}).
\end{aligned}  
\end{equation}
Since \(r_i(\mu)\leq 0\), we have
\begin{equation}\label{eqn: pf eqn2}
    \begin{aligned}
       & \int_{\mathbb{X}} f_i(x_i,x_{-i})\,d\mu(x_i,x_{-i})\\
&\leq \underset{y_i\in\mathbb{X}_i}{\min}\int_{\mathbb{X}} f_i(y_i,x_{-i}) \,d\mu(x_i,x_{-i})\\
&\leq  \int_{\mathbb{X}} f_i\Bigl(\bar{x}_i(x_{-i}),x_{-i}\Bigr)\,d\mu(x_i,x_{-i}).
    \end{aligned}
\end{equation}
Here, the second step in \eqref{eqn: pf eqn2} is due to the definition of minimum. Since the choice of \(\sigma_i\) is arbitrary,  \eqref{eqn: pf eqn1} and \eqref{eqn: pf eqn2} together implies \eqref{eqn: CE general}.
\end{proof}

\subsection{Correlated regret and approximate correlated equilibrium}
Since set \(\mathbb{X}\) contains infinitely many joint actions, it takes infinitely many parameters to characterize a general correlated equilibrium (a probably measure over set \(\mathbb{X}\)), which is intractable for computation and learning. To address this challenge, we propose to focus on the cases where \(\mu\) is the weighted sum of finitely many Dirac functions. In particular, let \(\hat{x}^1, \hat{x}^2, \ldots, \hat{x}^N\in\mathbb{X}\), denote a set of \emph{basis joint actions}. For \(k=1, 2, \ldots, N\), let
\begin{equation}
    \delta_k(x)\coloneqq \begin{cases}
    1, & x=\hat{x}^k,\\
    0, & \text{otherwise,}
    \end{cases}
\end{equation}
denote the Dirac function associated with action \(\hat{x}^k\in\mathbb{X}\). We consider joint action distributions of the following form
\begin{equation}\label{eqn: pdf approx}
    \mu(x)= \sum_{k=1}^N [w]_k \delta_k(x), 
\end{equation}
where \(w\in\Delta_N\coloneqq \{u\in\mathbb{R}_{\geq 0}^N| u^\top \mathbf{1}_N=1\}\), and \([w]_k\) is the \(k\)-the element of vector \(w\). When combined with \eqref{eqn: pdf approx}, the integrals in \eqref{eqn: correlated regret} becomes summation in this case. In particular, the correlated regret reduces to the folowing form
\begin{equation}\label{eqn: approx regret}
     \hat{r}_i(w)\coloneqq  \sum_{k=1}^N [w]_k f_i(\hat{x}_i^k, \hat{x}_{-i}^k)-\underset{y_i\in\mathbb{X}_i}{\min} \sum_{k=1}^N [w]_k f_i(y_i, \hat{x}_{-i}^k).
\end{equation}
In other words, \(\hat{r}_i(w)=r_i(\mu(x))\) where \(\mu(x)\) is given by \eqref{eqn: pdf approx}. 

Unlike the general case, the joint action distributions of the form \eqref{eqn: pdf approx} only requires finitely many parameters to describe, assuming that the set of basis joint actions is known. This parameterization allows efficient computation, but may no longer satisfy the conditions in Theorem~\ref{thm: equivalence-cost} exactly. Therefore, we introduce the following notion of \emph{approximate correlated equilibrium}.
 
\begin{definition}[Approximate Correlated Equilibrium]\label{def: approx CE}
  Consider a \(m\)-player convex game \(\{\mathbb{X}_i,f_i\}_{i=1}^m\) and basis actions \(\{\hat{x}^k\}_{k=1}^N\) where \(\hat{x}^k\in\mathbb{X}\) for all \(k=1, 2, \ldots, N\). For any \(w\in\Delta_N\), we let \(\hat{r}_{\text{avg}}(w)\coloneqq \frac{1}{m} \sum_{i=1}^m \hat{r}_i(w)\),
  where \(\hat{r}_i\)is given by \eqref{eqn: approx regret}. We say \(\mu(x)=\sum_{k=1}^N [w^\star]_k \delta_k(x)\) is an approximate correlated equilibrium  associated with basis actions \(\{\hat{x}^k\}_{k=1}^N\) if 
  \begin{equation}\label{eqn: approx CE conditions}
      w^\star\in\underset{w\in\Delta_N}{\argmin} \enskip\hat{r}_{\text{avg}}(w).
  \end{equation}
\end{definition}
Instead of the conditions in Theorem~\ref{thm: equivalence-cost} that require all regrets to be nonnegative, the conditions in \eqref{eqn: approx CE conditions} only require that the average value of \(\hat{r}_i(w)\) is minimized. Since the definition of \(\hat{r}_i(w)\) in \eqref{eqn: approx regret} implies that \(\hat{r}_i(w)\) is always nonnegative, the conditions in \eqref{eqn: approx CE conditions} implies that \(\hat{r}_i(w)\) is close to zero on average, which approximates the conditions in Theorem~\ref{thm: equivalence-cost}.  
   
We demonstrate the convergence of the players' regrets when solving the  

\section{Actively Learning Approximate Correlated Equilibrium via Correlated Regret}
\label{sec: active learning}
We now develop an active learning framework for computing approximate correlated equilibria based on correlated regrets. This framework combines Bayesian optimization with a heuristics method to identify a representative set of basis actions, as we will discuss below.  

\subsection{Online learning of approximate correlated equilibrium}
\label{subsec: problem formulation}

We consider the problem of learning a approximated correlated equilibrium based on players’ approximate correlated regret. For a given set of basis actions \(\{\hat{x}^k\}_{k=1}^N\) and a weighting vector \(w\in\Delta_N\), the correlated regret defined in \eqref{eqn: approx regret} quantifies the cost improvement that player \(i\) could achieve by unilaterally deviating from the action recommended according to \(\mu(x)\), which is defined in \eqref{eqn: pdf approx}. Our goal is to design a query-efficient learning procedure that estimates \(w^\star\) that satisfies \eqref{def: approx CE} by iteratively refining the joint distribution based on reported regrets. We summarize this problem below.

\begin{probdef*}
Let \(\{\hat{x}^k\}_{k=1}^N\) be a set of basis joints actions such that \(\hat{x}^k\in\mathbb{X}\) for all \(k=1, 2, \ldots, N\). Suppose that there exists an oracle such that, for any weight vector \(w\in\Delta_N\), returns the average correlated regret \(\hat{r}_{\text{avg}}(w)\). The goal is to design a sequential querying procedure that, in the \((n+1)\)-th round of query (\(n\in\mathbb{N}\)), selects \(w^n\in\Delta_N\) based on the history \(\mathcal{D}^n\coloneqq \{w^l, \hat{r}_{\text{avg}}(w^l)\}_{l=1}^n\), queries the oracle to return the average correlated regret \(\hat{r}_{\text{avg}}(w^n)\), and returns an estimate of \(w^\star\in\Delta_N\), such that \(\mu(x)=\sum_{k=1}^N [w^\star]_k \delta_k(x)\) is an approximate correlated equilibrium introduced in Definition~\ref{def: approx CE}.
\end{probdef*}

Learning from reported correlated regret is reasonable in settings where a central coordinator can propose joint action distributions and receive players' feedback. In multiplayer systems, each player can locally evaluate the cost reduction it could achieve by deviating from the recommended actions. Reporting this cost reduction allows the coordinator to assess the quality of its current recommendations without knowledge of players' cost functions. Since this feedback is evaluated locally and communicated privately to the coordinator---rather than to other potentially noncooperative players---it provides a lightweight, privacy-preserving mechanism for learning approximate correlated equilibria.
\subsection{Actively learning equilibrium via Bayesian optimization}
\label{subsec: Bayesian opt}
We propose a Bayesian optimization framework  \cite{shahriari2015taking} for solving the problem formulated in Section~\ref{subsec: problem formulation}. Bayesian optimization is a sequential strategy for optimizing black-box functions, particularly effective in reducing the number of objective function evaluations. This feature is important in our problem setting, since each evaluation requires querying all players' regrets for a given weight vector \(w\). 
This is achieved through two core components: a probabilistic surrogate model to capture our beliefs about the objective function, and an acquisition function to select the next query point by balancing exploration and exploitation. 

When applied to the problem formulated in Section~\ref{subsec: problem formulation}, the key idea of Bayesian optimization is to construct a probabilistic surrogate model that represents our belief about how different weight vector choices (\(w\)) influence the correlated regrets. We leverage this probabilistic model to define an acquisition function that quantifies the utility of evaluating the objective at different weight vectors by balancing {\em exploitation} of regions where the surrogate predicts low regret values against {\em exploration} of regions with high predictive uncertainty. Concretely, we build a Gaussian Process model \cite{williams1995gaussian} to predict the mean and variance of the approximate correlated regret defined in Definition~\ref{def: approx CE} as a function of the weight vector as inputs. 
A Gaussian Process model is characterized by the choice of a positive definite kernel function \(\kappa:\Delta_N\times \Delta_N\to\mathbb{R}\). After \(n\) queries, we have a dataset of observations \(\mathcal{D}^n\coloneqq \{w^l, \hat{r}_{\text{avg}}(w^l)\}_{l=1}^n\).  The GP model uses this data to form a posterior predictive distribution for the regret at any candidate weight vector $w$. This posterior is a Gaussian distribution, $p(\hat{r}_{\text{avg}}(w) | \mathcal{D}_n) = \mathcal{N}(\theta_n(w), \rho_n(w)^2)$, where $\theta_n(w)$ is the predicted mean and $\rho_n(w)^2$ is the predictive variance, which quantifies our uncertainty. The posterior mean and variance are as follows:
\begin{equation}
    \begin{aligned} 
    \theta_n(w) &:= c_n(w)^\top (K_n + \sigma^2 I_n)^{-1} \eta_n, \\
    \rho_n(w)^2 &:= \kappa(w, w) - c_n(w)^\top (K_n + \sigma^2 I_n)^{-1} c_n(w), \label{eqn: GP mean & var}
\end{aligned}
\end{equation}
where:
\begin{itemize}
    \item $\eta_n := [\hat{r}_{\text{avg}}(w^1), \hat{r}_{\text{avg}}(w^2), \dots, \hat{r}_{\text{avg}}(w^n)]^\top$ is the vector of observed regret values.
    \item $K_n \in \mathbb{R}^{n \times n}$ is the kernel matrix computed from the observed inputs, with $[K_n]_{ij} = \kappa(w^i, w^j)$.
    \item $c_n(w) := [\kappa(w^1, w), \kappa(w^2, w), \dots, \kappa(w^n, w)]^\top$ is the vector of covariances between the candidate point $w$ and the observed points.
    \item $\sigma\in\mathbb{R}_{>0}$ is a model hyperparameter.
\end{itemize}


The GP model, with its estimates of both mean \(\theta_n(w)\)  and uncertainty  \(\rho_n(w)\), provides the foundation for the acquisition function to intelligently guide the search. Concretely, equipped with this prediction, Bayesian optimization chooses the next query point \(w^{n+1}\in\Delta_N\) as follows
\begin{equation}
    w^{n+1}\in\underset{w\in\Delta_N}{\argmax}\enskip \psi(\theta_n(w), \rho_n(w))
\end{equation}
function \(\psi:\mathbb{R}\times \mathbb{R}\to\mathbb{R}\) is an \emph{acquisition function}. 
Concretely, we use Expected Improvement \cite{bull2011convergence} as acquisition function $\psi$ that maps the mean and variance of the Gaussian Process to the expected amount by which evaluating \(w^{n+1}\) would improve upon the current best observed regret value.  By iteratively maximizing the acquisition function and updating the GP with the new observation, the algorithm efficiently navigates the search space to find a weight vector $w^*$ that minimizes the average correlated regret.
\subsection{Identifying basis actions via maximizing differences}
\label{subsec: max diff}
The quality of the approximation in \eqref{eqn: pdf approx} depends heavily on the choice of basis actions \(\{\hat{x}^k\}_{k=1}^N\). Intuitively, we want these basis actions to be as diverse as possible, such that they provide representative samples of all possible joint actions. To this end,
we propose the following optimization problem to identify a diverse set of joint actions:
\begin{equation}\label{opt: max diff}
    \begin{array}{ll}
       \underset{\{\hat{x}^k\}_{k=1}^N}{\mbox{maximize}}  & \underset{(i, j)\in\mathbb{P}}{\min} \psi(\hat{x}^i-\hat{x}^j)  \\
      \mbox{subject to}   & \hat{x}^k\in\mathbb{X}, \enskip k=1, 2, \ldots, N, 
    \end{array}
\end{equation} 
where \(\psi:\mathbb{R}^{\sum_{i=1}^m n_i}\to\mathbb{R}_{\geq 0}\) is a nonnegative-valued convex  function that measures differences (\eg, \(\ell_p\) norms), and 
\begin{equation}
    \mathbb{P}\coloneqq \{(i, j)| 1\leq i, j\leq N, i\neq j\}.
\end{equation}
The idea is to maximize pairwise distances among all basis actions. 
The optimization in \eqref{opt: max diff} is equivalent to a difference-of-convex programming \cite{horst1999dc}. To see this equivalence, notice that \(\underset{(i, j)\in\mathbb{P}}{\min} \psi(\hat{x}^i-\hat{x}^j) = \sum_{(i, j)\in\mathbb{P}} \psi\left(\hat{x}^i-\hat{x}^j\right)-\underset{(i, j)\in\mathbb{P}}{\mbox{max}} \sum_{(p, q)\in\mathbb{P}\setminus \{(i, j)\}} \psi\left(\hat{x}^p-\hat{x}^q\right)\).
Hence, by adding a slack variable \(s\), we can rewrite optimization~\eqref{opt: max diff} as follows
\begin{equation}\label{opt: dcp}
    \begin{array}{ll}
       \underset{s, \{\hat{x}^k\}_{k=1}^N}{\mbox{maximize}}  &  -s+\sum\limits_{(i, j)\in\mathbb{P}} \psi(\hat{x}^i-\hat{x}^j) \\
      \mbox{subject to}  &  \hat{x}^k\in\mathbb{X}, \enskip k=1, 2, \ldots, N,\\
      & \sum\limits_{(p, q)\in\mathbb{P}\setminus \{(i, j)\}} \psi\left(\hat{x}^p-\hat{x}^q\right)\leq s, \enskip (i, j)\in \mathbb{P}.
    \end{array}
\end{equation}
Therefore, the optimization in \eqref{opt: dcp} is a difference-of-convex program: it contains convex constraints, but its objective function is the difference between two convex functions \cite{horst1999dc}. A popular solution method for difference-of-convex programs is the \emph{convex-concave procedure}, which guarantees global convergence to a stationary point and provides locally optimal solutions  in practice  \cite{lipp2016variations}.  

\section{Numerical Experiments}
We demonstrate the active learning framework proposed in Section~\ref{sec: active learning} via a multi-player traffic assignment game defined over the Sioux Falls network, a common benchmark network model. All network data used in this section come from the \emph{Transportation Networks for Research} \cite{transpnet}.

\subsection{Multi-player traffic assignment game}
\label{subsec: SF net}

\begin{wrapfigure}{r}{0.5\linewidth} 
    \centering
    \includegraphics[width=\linewidth, trim={9cm 2cm 8cm 1cm}, clip]{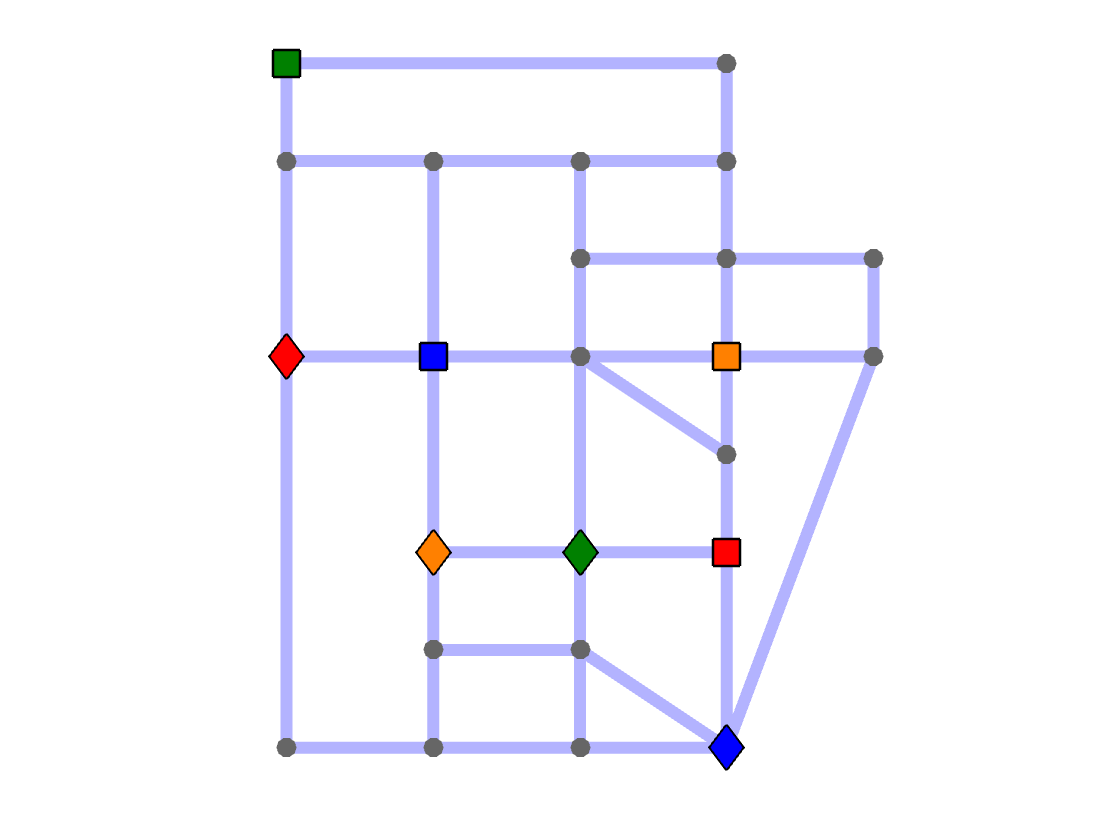}

    \caption{Origin and destination nodes of four players in the Sioux Falls network. Square and diamond markers illustrate origin and destination nodes, respectively. Blue, green, red, and orange markers correspond to player 1 to player 4, respectively.}
    \label{fig: SFnet}
\end{wrapfigure}
We consider a multi-player traffic assignment game in the Sioux Falls network, which contains 24 nodes and 76 links. Each player's objective is to assign traffic flows to the links in the network, such that the assignment serves a given amount of demand from some origin node to a destination node. Fig~\ref{fig: SFnet} illustrate the structure of the Sioux Falls network, as well as the origin and destination nodes of four players. Each player aims to optimize its cost function, which is different from other player's cost functions. The value of this function depends on the traffic assignment of other players, since they create congestion on the links and, as a result, may affect the value of each player's cost function.

To  model this problem, we let \(\mathcal{N}=\{1, 2, \ldots, n_n\}\) denote the set of nodes, and \(\mathcal{L}=\{1, 2, \ldots, n_l\}\) denote the set of links. Each link is an ordered pair of distinct nodes, where the first node and second node are the ``tail" and ``head" of the link, respectively. We represent the topology of the hybrid air-ground network using the node-link incidence matrix \(E\in\mathbb{R}^{n_n\times n_l}\). The entry \([E]_{ik}\) equals \(1\) if node \(i\) is the tail of link \(k\), \(-1\) if node \(i\) is the head of link \(k\), and zero otherwise. We let \(a, b\in\mathbb{R}^{n_l}_{>0}\) denote the cost parameters associated with player \(i\), where the \(j\)-th element of \(a\) and \(b\)---denoted by \([a]_j\) and \([b]_j\), respectively---denote the nominal travel time and nominal traffic volume of link \(j\).

We consider \(m\in\mathbb{N}\) players of traffic routed through this network. Player \(i\)'s goal is transport a given amount of traffic demands from its origin node \(o_i\in\mathbb{N}\) to its destination node \(d_i\in\mathbb{N}\). We let \(\rho_i\) denote the amount of traffic demand for the \(i\)-th player, and \(s_i\in\mathbb{R}^{n_n}\) denote the \emph{demand vector} for player \(i\), such that the \(j\)-th element of \(s_i\) is \(\rho_i\) if \(j=o_i\), \(-\rho_i\) if \(j=d_i\), and zero otherwise. We let 
\begin{equation}\label{eqn: constraint set}
 \mathbb{X}_i\coloneqq \{x_i\in\mathbb{R}^{n_l}| Ex_i=s_i, 0_{n_l}\leq x_i\leq c, a^\top x\leq \gamma_i\}
\end{equation}
denote the feasible set of flow vectors for player \(i\), where \(c\in\mathbb{R}^{n_l}\) is the capacity vector, \(\gamma_i\in\mathbb{R}_{\geq  0}\) is an upper bound on player \(i\)'s cost when the travel time on each link equals its nominal value. Furthermore, each player evaluates the cost of using a link based on a BPR-type cost function \cite{patriksson2015traffic}.  Let \(x_i\) and \(x_{-i}\) denote the flow vectors of player \(i\) and players other than \(i\), respectively. We let \(\ell_j:\mathbb{X}_i\times \mathbb{X}_{-i}\to\mathbb{R}\) denote the BPR-type cost function for link \(i\in\mathcal{L}\), such that
\begin{equation}\label{eqn: BPR cost}
   \textstyle \ell_i(x_i, x_{-i})\coloneqq [a_i]_j  \left(1+\lambda \left(\frac{\sum_{i=1}^m \left[x_i\right]_j}{[b_i]_j}\right)^\nu\right)
\end{equation}
where \([a_i]_j\) and \([b_i]_j\) denote the \(j\)-th entry of vector \(a_i\) and \(b_i\), respectively; \(\lambda\in\mathbb{R}_{>0}\), \(\nu\in\mathbb{N}\) are model parameters of the BPR model. The value of \(\left[\sum_{i=1}^m x_i\right]_j\) denotes the total amount of traffic on link \(j\) assigned by all players. A typical choice of parameter value is \(\lambda=0.15, \nu=4\). This function models how the cost of using a link increases with the total amount of traffic contributed by all players. 

Based on this link cost function, we let \(f_i:\mathbb{X}_i\times \mathbb{X}_{-i}\to\mathbb{R}\) be as follows
\begin{equation}
   \textstyle f_i(x_i, x_{-i})=\frac{1}{\delta_i}\sum_{j=1}^{n_l} [x_i]_j \ell_j(x_i, x_{-i}) 
\end{equation}
where \(\delta_i\coloneqq \min_{x_i\in\mathbb{X}_i} a^\top x_i\) is the nominal cost value for player \(i\) when there is no traffic in the network. Intuitively, function \(f_i\) evaluates a weighted sum of the flow assigned by player \(i\), where the flow on link \(j\) is weighted by the BPR-type link cost of link \(j\). Notice that although \(f_i(x_i, x_{-i})\) can be an odd order polynomial of the elements of \(x_i\), it is a convex function for any \(x_i\in\mathbb{R}_{\geq 0}^{n_l}\). One can verify this fact by checking that the Hessian matrix \(\nabla^2_{x_i} f_i(x_i, x_{-i})\) is positive semidefinite if \(x_j\in\mathbb{R}_{\geq 0}^{n_l}\) for all \(j=1, 2, \ldots, m\).   

\subsection{Learning traffic assignments from correlated regrets}
We demonstrate the proposed active learning framework using a two-player (player 1 and player 2) and a four-player case of the multi-player traffic assignment game illustrated by Fig.~\ref{fig: SFnet}. For all the parameters used in Section~\ref{subsec: SF net} (including the node-link incidence matrix, traffic demands for each origin-destination pair, as well as the parameters of the BPR-type cost function in \eqref{eqn: BPR cost}), we set their values based on the dataset developed by the Transportation Networks for Research Core Team \cite{transpnet}. We first identify a set of basis actions using the heuristic method proposed in Section~\ref{subsec: max diff}, where we choose function \(\psi\) as the \(\ell_1\)-norm function and set \(N=10\) and \(N=15\) for the two-player and four-player case, respectively. We then apply Bayesian optimization method discussed in Section~\ref{subsec: Bayesian opt} to learn an approximate correlated equilibrium of the multi-player traffic assignment game. In the Bayesian optimization method, we choose the \emph{Mat\'ern kernel} with \(\nu=5/2\), and the \emph{Expected Improvement} as the acquisition function, and consider the explicit constraints of \(w\in\Delta_N\) in the acquisition optimization step. We base our implementation of the Bayesian optimization method on the open-source toolbox BoTorch \cite{balandat2020botorch}. 

Figure~\ref{fig: regret} illustrate the convergence of players' correlated regrets---which is defined by \eqref{eqn: approx regret}---when applying the Bayesian optimization method proposed in Section~\ref{subsec: Bayesian opt}. As a benchmark to the heuristic for identifying basis actions introduce in Section~\ref{subsec: max diff}, we also test the proposed Bayesian optimization method based on random basis actions. We generate each random joint action by optimizing a random linear function---whose coefficients are uniformly sampled from the interval \([0, 1]\)---over the joint action set \(\mathbb{X}\). The results in Figure~\ref{fig: regret} show that the proposed Bayesian optimization method consistently reduces the payers correlated regrets over relatively few iterations. This showcases the advantage of Bayesian optimization over other black-box optimization methods: it is effective in reducing the number of queries of the objective function value (correlated regrets in this case). Furthermore, the results in Fig.~\ref{fig: regret} also show that using the basis actions identified by heuristic in Section~\ref{subsec: max diff} outperforms those based on random basis actions, by approximately one order of magnitude in the four-player game. This improvement indicates that the proposed heuristic leads to a diverse set of actions that better captures the overall structure of joint action set \(\mathbb{X}\).

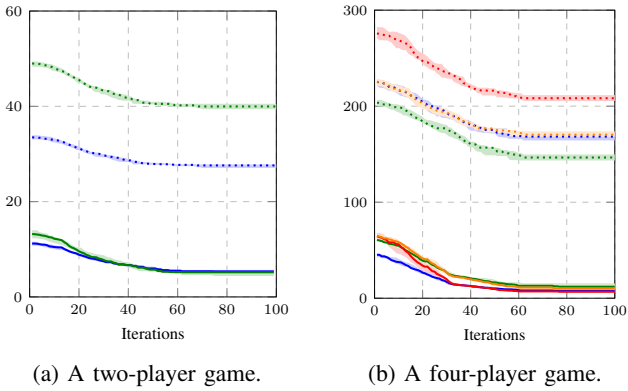
\begin{figure}[t]
  \centering
  \begin{subfigure}[t]{0.48\linewidth}
    \centering
    \resizebox{\textwidth}{!}{\pgfplotsset{every tick label/.append style={font=\scriptsize}}
\begin{tikzpicture}
\begin{axis}[
width=1\columnwidth,      
  height=1.16\columnwidth,      
  scale only axis,  
xmin=0, xmax=100,
ymin=0, ymax=60,
grid=major,               
    major grid style={dashed,gray!50}, 
    xlabel = \footnotesize Iterations,
    xshift = 0pt
]

\definecolor{player1color}{rgb}{0,0,1}
\definecolor{player2color}{rgb}{0,0.5,0}

\addplot [mark=none, color=player1color, line width=1pt, dotted] table[x=iteration,y=player1_median] {ccp_vs_no_ccp_normalized_for_two_players_no_ccp_normalized.dat};

\addplot [mark=none, color=player2color, line width=1pt, dotted] table[x=iteration,y=player2_median] {ccp_vs_no_ccp_normalized_for_two_players_no_ccp_normalized.dat};

\addplot [mark=none, color=player1color, line width=1pt, solid] table[x=iteration,y=player1_median] {ccp_vs_no_ccp_normalized_for_two_players_ccp_normalized.dat};

\addplot [mark=none, color=player2color, line width=1pt, solid] table[x=iteration,y=player2_median] {ccp_vs_no_ccp_normalized_for_two_players_ccp_normalized.dat};

\addplot [name path=upper0p1, draw=none] table[x=iteration,y=player1_q3] {ccp_vs_no_ccp_normalized_for_two_players_no_ccp_normalized.dat};
\addplot [name path=lower0p1, draw=none] table[x=iteration,y=player1_q1] {ccp_vs_no_ccp_normalized_for_two_players_no_ccp_normalized.dat};
\addplot [fill=player1color!20] fill between[of=upper0p1 and lower0p1];

\addplot [name path=upper0p2, draw=none] table[x=iteration,y=player2_q3] {ccp_vs_no_ccp_normalized_for_two_players_no_ccp_normalized.dat};
\addplot [name path=lower0p2, draw=none] table[x=iteration,y=player2_q1] {ccp_vs_no_ccp_normalized_for_two_players_no_ccp_normalized.dat};
\addplot [fill=player2color!20] fill between[of=upper0p2 and lower0p2];

\addplot [name path=upper1p1, draw=none] table[x=iteration,y=player1_q3] {ccp_vs_no_ccp_normalized_for_two_players_ccp_normalized.dat};
\addplot [name path=lower1p1, draw=none] table[x=iteration,y=player1_q1] {ccp_vs_no_ccp_normalized_for_two_players_ccp_normalized.dat};
\addplot [fill=player1color!20] fill between[of=upper1p1 and lower1p1];

\addplot [name path=upper1p2, draw=none] table[x=iteration,y=player2_q3] {ccp_vs_no_ccp_normalized_for_two_players_ccp_normalized.dat};
\addplot [name path=lower1p2, draw=none] table[x=iteration,y=player2_q1] {ccp_vs_no_ccp_normalized_for_two_players_ccp_normalized.dat};
\addplot [fill=player2color!20] fill between[of=upper1p2 and lower1p2];

\end{axis}
\end{tikzpicture}}
    \caption{A two-player game.}
    \label{fig: two players}
  \end{subfigure}\hfill
  \begin{subfigure}[t]{0.48\linewidth}
    \centering
   \resizebox{\textwidth}{!}{\pgfplotsset{every tick label/.append style={font=\scriptsize}}
\begin{tikzpicture}
\begin{axis}[
width=1\columnwidth,      
  height=1.2\columnwidth,      
  scale only axis,  
xmin=0, xmax=100,
ymin=0, ymax=300,
grid=major,               
    major grid style={dashed,gray!50}, 
    xlabel = \footnotesize Iterations,
    xshift = 0pt
]

\definecolor{player1color}{rgb}{0,0,1}
\definecolor{player2color}{rgb}{0,0.5,0}
\definecolor{player3color}{rgb}{1,0,0}
\definecolor{player4color}{rgb}{1,0.5,0}

\addplot [mark=none, color=player1color, line width=1pt, solid] table[x=iteration,y=player1_median] {ccp_vs_no_ccp_normalized_no_ccp_normalized_average_case.dat};

\addplot [mark=none, color=player2color, line width=1pt, solid] table[x=iteration,y=player2_median] {ccp_vs_no_ccp_normalized_no_ccp_normalized_average_case.dat};

\addplot [mark=none, color=player3color, line width=1pt, solid] table[x=iteration,y=player3_median] {ccp_vs_no_ccp_normalized_no_ccp_normalized_average_case.dat};

\addplot [mark=none, color=player4color, line width=1pt, solid] table[x=iteration,y=player4_median] {ccp_vs_no_ccp_normalized_no_ccp_normalized_average_case.dat};

\addplot [mark=none, color=player1color, line width=1pt, dotted] table[x=iteration,y=player1_median] {ccp_vs_no_ccp_normalized_ccp_normalized_average_case.dat};

\addplot [mark=none, color=player2color, line width=1pt, dotted] table[x=iteration,y=player2_median] {ccp_vs_no_ccp_normalized_ccp_normalized_average_case.dat};

\addplot [mark=none, color=player3color, line width=1pt, dotted] table[x=iteration,y=player3_median] {ccp_vs_no_ccp_normalized_ccp_normalized_average_case.dat};

\addplot [mark=none, color=player4color, line width=1pt, dotted] table[x=iteration,y=player4_median] {ccp_vs_no_ccp_normalized_ccp_normalized_average_case.dat};

\addplot [name path=upper0p1, draw=none] table[x=iteration,y=player1_q3] {ccp_vs_no_ccp_normalized_no_ccp_normalized_average_case.dat};
\addplot [name path=lower0p1, draw=none] table[x=iteration,y=player1_q1] {ccp_vs_no_ccp_normalized_no_ccp_normalized_average_case.dat};
\addplot [fill=player1color!20] fill between[of=upper0p1 and lower0p1];

\addplot [name path=upper0p2, draw=none] table[x=iteration,y=player2_q3] {ccp_vs_no_ccp_normalized_no_ccp_normalized_average_case.dat};
\addplot [name path=lower0p2, draw=none] table[x=iteration,y=player2_q1] {ccp_vs_no_ccp_normalized_no_ccp_normalized_average_case.dat};
\addplot [fill=player2color!20] fill between[of=upper0p2 and lower0p2];

\addplot [name path=upper0p3, draw=none] table[x=iteration,y=player3_q3] {ccp_vs_no_ccp_normalized_no_ccp_normalized_average_case.dat};
\addplot [name path=lower0p3, draw=none] table[x=iteration,y=player3_q1] {ccp_vs_no_ccp_normalized_no_ccp_normalized_average_case.dat};
\addplot [fill=player3color!20] fill between[of=upper0p3 and lower0p3];

\addplot [name path=upper0p4, draw=none] table[x=iteration,y=player4_q3] {ccp_vs_no_ccp_normalized_no_ccp_normalized_average_case.dat};
\addplot [name path=lower0p4, draw=none] table[x=iteration,y=player4_q1] {ccp_vs_no_ccp_normalized_no_ccp_normalized_average_case.dat};
\addplot [fill=player4color!20] fill between[of=upper0p4 and lower0p4];

\addplot [name path=upper1p1, draw=none] table[x=iteration,y=player1_q3] {ccp_vs_no_ccp_normalized_ccp_normalized_average_case.dat};
\addplot [name path=lower1p1, draw=none] table[x=iteration,y=player1_q1] {ccp_vs_no_ccp_normalized_ccp_normalized_average_case.dat};
\addplot [fill=player1color!20] fill between[of=upper1p1 and lower1p1];

\addplot [name path=upper1p2, draw=none] table[x=iteration,y=player2_q3] {ccp_vs_no_ccp_normalized_ccp_normalized_average_case.dat};
\addplot [name path=lower1p2, draw=none] table[x=iteration,y=player2_q1] {ccp_vs_no_ccp_normalized_ccp_normalized_average_case.dat};
\addplot [fill=player2color!20] fill between[of=upper1p2 and lower1p2];

\addplot [name path=upper1p3, draw=none] table[x=iteration,y=player3_q3] {ccp_vs_no_ccp_normalized_ccp_normalized_average_case.dat};
\addplot [name path=lower1p3, draw=none] table[x=iteration,y=player3_q1] {ccp_vs_no_ccp_normalized_ccp_normalized_average_case.dat};
\addplot [fill=player3color!20] fill between[of=upper1p3 and lower1p3];

\addplot [name path=upper1p4, draw=none] table[x=iteration,y=player4_q3] {ccp_vs_no_ccp_normalized_ccp_normalized_average_case.dat};
\addplot [name path=lower1p4, draw=none] table[x=iteration,y=player4_q1] {ccp_vs_no_ccp_normalized_ccp_normalized_average_case.dat};
\addplot [fill=player4color!20] fill between[of=upper1p4 and lower1p4];

\end{axis}
\end{tikzpicture}}
    \caption{A four-player game.}
    \label{fig: four players}
  \end{subfigure}

  \caption{Convergence of individual player's correlated regret in two-player and four-player traffic assignment games. Blue, green, red, and orange curves correspond to players 1–4. Each curve shows the regret averaged over 10 Monte Carlo simulations, with shaded regions indicating the corresponding ranges. The dotted and solid curves correspond to results based on random basis actions and on heuristic-based basis actions (Section~\ref{subsec: max diff}), respectively.}
    \label{fig: regret}
\end{figure}


\section{Conclusion}
We proposed an active learning framework for learning approximate correlated equilibrium in convex games by querying players' regrets. This framework combine Bayesian optimization with a heuristic that select a finite set of joint actions in an infinite joint action space by maximizing pairwise differences. This framework leads to efficient learning in multi-user traffic assignment games.

However, the current work still has several limitations. For example, it is unclear how the proposed heuristic method for selecting joint actions scale in high-dimensional dynamic games. It is also unclear how to theoretically bound the approximation error of the learned approximate correlated equilibrium. In our future work, we aim to address these limitations, as well as consider extending the current work to decentralized learning of multiple coordinators.  




\bibliographystyle{IEEEtran}
\bibliography{IEEEabrv,reference}

\begin{thebibliography}{10}
\providecommand{\url}[1]{#1}
\csname url@rmstyle\endcsname
\providecommand{\newblock}{\relax}
\providecommand{\bibinfo}[2]{#2}
\providecommand\BIBentrySTDinterwordspacing{\spaceskip=0pt\relax}
\providecommand\BIBentryALTinterwordstretchfactor{4}
\providecommand\BIBentryALTinterwordspacing{\spaceskip=\fontdimen2\font plus
\BIBentryALTinterwordstretchfactor\fontdimen3\font minus \fontdimen4\font\relax}
\providecommand\BIBforeignlanguage[2]{{%
\expandafter\ifx\csname l@#1\endcsname\relax
\typeout{** WARNING: IEEEtran.bst: No hyphenation pattern has been}%
\typeout{** loaded for the language `#1'. Using the pattern for}%
\typeout{** the default language instead.}%
\else
\language=\csname l@#1\endcsname
\fi
#2}}

\bibitem{aumann1974subjectivity}
R.~J. Aumann, ``Subjectivity and correlation in randomized strategies,'' \emph{J. Math. Econ.}, vol.~1, no.~1, pp. 67--96, 1974.

\bibitem{aumann1987correlated}
------, ``Correlated equilibrium as an expression of bayesian rationality,'' \emph{Econometrica}, pp. 1--18, 1987.

\bibitem{hart1989existence}
S.~Hart and D.~Schmeidler, ``Existence of correlated equilibria,'' \emph{Math. Oper. Res.}, vol.~14, no.~1, pp. 18--25, 1989.

\bibitem{cavaliere2001coordination}
A.~Cavaliere, ``Coordination and the provision of discrete public goods by correlated equilibria,'' \emph{J. Public Econ. Theory}, vol.~3, no.~3, pp. 235--255, 2001.

\bibitem{babichenko2014simple}
Y.~Babichenko, S.~Barman, and R.~Peretz, ``Simple approximate equilibria in large games,'' in \emph{Proc. 15th ACM Conf. Econ. Comput.}, 2014, pp. 753--770.

\bibitem{duffy2017coordination}
J.~Duffy, E.~K. Lai, and W.~Lim, ``Coordination via correlation: An experimental study,'' \emph{Econ. Theory}, vol.~64, no.~2, pp. 265--304, 2017.

\bibitem{wu2012cooperative}
D.~Wu, Y.~Cai, L.~Zhou, Z.~Zheng, and B.~Zheng, ``Cooperative strategies for energy-aware ad hoc networks: A correlated-equilibrium game-theoretical approach,'' \emph{IEEE Trans. Veh. Technol.}, vol.~62, no.~5, pp. 2303--2314, 2012.

\bibitem{papadimitriou2008computing}
C.~H. Papadimitriou and T.~Roughgarden, ``Computing correlated equilibria in multi-player games,'' \emph{J. ACM}, vol.~55, no.~3, pp. 1--29, 2008.

\bibitem{kamisetty2011approximating}
H.~Kamisetty, E.~P. Xing, and C.~J. Langmead, ``Approximating correlated equilibria using relaxations on the marginal polytope.'' in \emph{ICML}, 2011, pp. 1153--1160.

\bibitem{im2024coordination}
J.~Im, Y.~Yu, D.~Fridovich-Keil, and U.~Topcu, ``Coordination in noncooperative multiplayer matrix games via reduced rank correlated equilibria,'' \emph{IEEE Control Syst. Lett.}, vol.~8, pp. 1637--1642, 2024.

\bibitem{hart2000simple}
S.~Hart and A.~Mas-Colell, ``A simple adaptive procedure leading to correlated equilibrium,'' \emph{Econometrica}, vol.~68, no.~5, pp. 1127--1150, 2000.

\bibitem{foster1997calibrated}
D.~P. Foster and R.~V. Vohra, ``Calibrated learning and correlated equilibrium,'' \emph{Games Econ. Behav.}, vol.~21, no. 589, pp. 40--55, 1997.

\bibitem{borowski2014learning}
H.~P. Borowski, J.~R. Marden, and J.~S. Shamma, ``Learning efficient correlated equilibria,'' in \emph{Proc. 53rd IEEE Conf. Decis. Control}.\hskip 1em plus 0.5em minus 0.4em\relax IEEE, 2014, pp. 6836--6841.

\bibitem{marden2017selecting}
J.~R. Marden, ``Selecting efficient correlated equilibria through distributed learning,'' \emph{Games Econ. Behav.}, vol. 106, pp. 114--133, 2017.

\bibitem{arifovic2019learning}
J.~Arifovic, J.~F. Boitnott, and J.~Duffy, ``Learning correlated equilibria: An evolutionary approach,'' \emph{J. Econ. Behav. Organ.}, vol. 157, pp. 171--190, 2019.

\bibitem{shahriari2015taking}
B.~Shahriari, K.~Swersky, Z.~Wang, R.~P. Adams, and N.~De~Freitas, ``Taking the human out of the loop: A review of {B}ayesian optimization,'' \emph{Proceedings of the IEEE}, vol. 104, no.~1, pp. 148--175, 2015.

\bibitem{williams1995gaussian}
C.~Williams and C.~Rasmussen, ``Gaussian processes for regression,'' \emph{Adv. Neural Inform. Process. Syst.}, vol.~8, 1995.

\bibitem{bull2011convergence}
A.~D. Bull, ``Convergence rates of efficient global optimization algorithms.'' \emph{J. Mach. Learn. Res.}, vol.~12, no.~10, 2011.

\bibitem{horst1999dc}
R.~Horst and N.~V. Thoai, ``{DC} programming: overview,'' \emph{J. Optim. Theory Appl.}, vol. 103, pp. 1--43, 1999.

\bibitem{lipp2016variations}
T.~Lipp and S.~Boyd, ``Variations and extension of the convex--concave procedure,'' \emph{Optim. Eng.}, vol.~17, pp. 263--287, 2016.

\bibitem{transpnet}
{Transportation Networks for Research Core Team}, ``Transportation networks for research,'' \url{https://github.com/bstabler/TransportationNetworks}, 2016, accessed: 2025-06-01.

\bibitem{patriksson2015traffic}
M.~Patriksson, \emph{The traffic assignment problem: models and methods}.\hskip 1em plus 0.5em minus 0.4em\relax Courier {D}over {P}ublications, 2015.

\bibitem{balandat2020botorch}
M.~Balandat, B.~Karrer, D.~Jiang, S.~Daulton, B.~Letham, A.~G. Wilson, and E.~Bakshy, ``Botorch: A framework for efficient monte-carlo bayesian optimization,'' \emph{Adv. Neural Inform. Process. Syst.}, vol.~33, pp. 21\,524--21\,538, 2020.

\end{thebibliography}

\end{document}